\documentclass[letterpaper,10pt]{article}

\usepackage[svgnames]{xcolor}
\usepackage[numbers,longnamesfirst]{natbib}

\usepackage{xstring}
\usepackage{amsmath}
\usepackage{amssymb}
\usepackage{amsthm}

\usepackage{mathtools}
\usepackage{url}
\usepackage{hyperref}

\usepackage[legalpaper, margin=1.25in]{geometry}

\def\titl{Fast interpolation and multiplication of unbalanced polynomials}
\def\pdftitle{\titl}
\def\auts{Pascal Giorgi, Bruno Grenet, Armelle Perret du Cray, Daniel S. Roche}

\hypersetup{
  pdfauthor={\auts},
  pdftitle={\pdftitle},
  colorlinks,
  linkcolor=DarkBlue,
  citecolor=DarkGreen,
  urlcolor=DarkBlue,
  bookmarksnumbered,
}

\usepackage{algorithm}
\usepackage[noend]{algpseudocode}
\algnewcommand\Input{\item[\textbf{Input:}]}
\algnewcommand\Output{\item[\textbf{Output:}]}
\algnewcommand\algorithmicforeach{\textbf{for each}}
\algdef{S}[FOR]{ForEach}[1]{\algorithmicforeach\ #1\ \algorithmicdo}
\usepackage{xspace}
\usepackage{enumitem}

\usepackage[nameinlink,capitalise]{cleveref}

\newtheorem{thm}{Theorem}[section]   
\newtheorem{lem}[thm]{Lemma}
\newtheorem{cor}[thm]{Corollary}
\newtheorem{fact}[thm]{Fact}

\newtheorem{definition}[thm]{Definition}

\theoremstyle{remark}

\newcommand{\gO}{O}
\newcommand{\gOt}{\tilde{O}}

\newcommand{\ff}[1]{\mathbb{F}_{\!#1}}
\newcommand{\fq}{\ff{q}}
\newcommand{\zz}{\mathbb{Z}}

\newcommand{\nn}{\mathbb{N}}
\newcommand{\bigo}[1]{\ensuremath{\gO\!\left(#1\right)}}
\newcommand{\bigot}[1]{\ensuremath{\gOt\!\left(#1\right)}}

\newcommand{\polylog}{\ensuremath{\mathsf{polylog}}}
\newcommand{\lb}{\log_2}

\newcommand{\ceil}[1]{\ensuremath{\lceil {#1} \rceil}}

\newcommand{\height}[1]{\ensuremath{\mathsf{height}\!\left(#1\right)}}
\newcommand{\abs}[1]{\ensuremath{\left|#1\right|}}
\newcommand{\round}[1]{\ensuremath{\left\lfloor#1\right\rceil}}

\DeclareMathOperator{\brem}{rem}
\DeclareMathOperator{\llog}{loglog}

\DeclareMathOperator{\IDFT}{IDFT}

\newcommand\algo[1]{\textsc{#1}\xspace}
\newcommand\Uinterpolate{\algo{Uinterpolate}}

\newcommand{\UintSlice}{\algo{UinterpolateSlice}}
\newcommand{\Uint}{\algo{Uinterpolate}}
\newcommand{\SupSup}{\algo{SupportSuperset}}
\newcommand{\esuper}{\ensuremath{\mathcal{T}}}
\newcommand{\VerifProd}{\algo{VerifProd}}
\newcommand{\UnbalancedProd}{\algo{UnbalancedProd}}

\newcommand{\bcinterp}{\algo{Interpolate\_mbb}}

\newcommand{\extoMDBB}{\algo{ExplicitMDBB}}
\newcommand{\MBBtoMDBB}{\algo{MBBtoMDBB}}
\newcommand{\sumMDBB}{\algo{SumMDBB}}
\newcommand{\prodMDBB}{\algo{ProdMDBB}}
\newcommand{\MBBimg}{\algo{MDBBImage}}

\newcommand{\bitlen}[1]{\ensuremath{\mathsf{bitlen}\!\left(#1\right)}}
\newcommand{\bitlenp}[1]{\ensuremath{\mathsf{bitlen}_+\!\left(#1\right)}}
\newcommand{\bitlenx}[1]{\ensuremath{\mathsf{bitlen}_x\!\left(#1\right)}}

\newcommand{\fs}{\ensuremath{f^*}}
\newcommand{\fss}{\ensuremath{f^{**}}}

\usepackage{xcolor}

\title{\titl}

\newcommand\auth[4]{%
  \begin{minipage}{.45\textwidth}%
  \centering
      #1\\%
      \normalsize
      #2\\%
      #3\\%
      #4
  \end{minipage}%
}

\author{%
  \auth{Pascal Giorgi}{LIRMM, Univ. Montpellier, CNRS}{Montpellier, France}{pascal.giorgi@lirmm.fr}
  \auth{Bruno Grenet}{LJK, Univ. Grenoble Alpes, CNRS}{Grenoble, France}{bruno.grenet@univ-grenoble-alpes.fr}\\[3em]
  \auth{Armelle Perret du Cray}{University of Waterloo}{Waterloo, ON, Canada}{aperretducray@uwaterloo.ca}
  \auth{Daniel S. Roche}{United States Naval Academy}{Annapolis, Maryland, U.S.A}{roche@usna.edu}
}

\begin{document}
\maketitle

\begin{abstract}

{We consider the classical problems of interpolating a polynomial given a black box for evaluation, and of
  multiplying two polynomials, in the setting where the bit-lengths of the coefficients may vary widely, so-called
  \emph{unbalanced} polynomials. Let $f\in\zz[x]$ be an unknown polynomial and $s,D$ be bounds on its total bit-length
  and degree, our new interpolation algorithm returns $f$ with high probability using $\bigot{s\log D}$ bit operations
  and $\bigo{s\log D \log s}$ black box evaluation. For polynomial multiplication, assuming the bit-length $s$ of the
  product is not given, our algorithm has an expected running time of $\bigot{s\log D}$, whereas previous methods for
  (resp.) dense or sparse arithmetic have at least $\bigot{sD}$ or $\bigot{s^2}$ bit complexity. }
\end{abstract}

\section{Introduction}

Consider a univariate polynomial with integer coefficients, written
\begin{equation}\label{eqn:ff}
f = c_1x^{e_1} + c_2x^{e_2} + \cdots + c_{t}x^{e_{t}},
\end{equation}
and write $H=\height{f}=\max_i\abs{c_i}$ and
$D=\deg f=\max_i e_i$.

Traditionally, fast algorithms have either used the \emph{dense representation} as a list of coefficients with total
size $\bigo{D\log H}$, or in the \emph{sparse representation} (a.k.a. \emph{supersparse} or \emph{lacunary}) as a list
of $t$ nonzero coefficient-exponent pairs with total size $\bigo{t\log (DH)}$. Considerable effort has been made since
the 1970s to develop quasi-linear algorithms for many problems with dense and sparse polynomials, but notice that in
either case, there is an implicit assumption that all coefficients have the same bit-size $\lb H$.

Instead, we focus on the \emph{total bit-length} of $f$, which
we write as $s = \sum_i(\lb\abs{c_i} + \lb e_i)$.
Note that both sizes above --- and therefore also any quasi-linear time
algorithms in those models --- can be as large as quadratic in $s$.
Our goal of quasi-linear complexity with respect to $s$
can be viewed as a natural progression: dense algorithms
use $\lb H$ bits to represent every possible coefficient;
sparse algorithms are more refined by avoiding storage of zero
coefficients; and \emph{unbalanced} algorithms (as we are proposing)
refine further by using the exact number of bits for every term.

\subsection{Our contributions}

  We provide new algorithms for interpolation and multiplication of
  polynomials with unbalanced coefficients:

  \begin{itemize}[leftmargin=*]
  \item (\cref{thm:Uint}) There is a Monte Carlo randomized algorithm that takes a \emph{modular black
      box} (see \cref{def:mbb}) for evaluating an unknown polynomial $f\in\zz[x]$, and bounds $D,s$ on its degree and
    total bit-size, and produces the sparse, unbalanced representation of $f$ with high probability. It uses
    $\bigo{s\log D\log s}$ black box evaluations and $\bigot{s\log D}$ additional bit operations.
  \item (\cref{thm:unbalancedproduct}) There is a randomized algorithm that takes two polynomials $f,g\in\zz[x]$ and
    with high probability produces their product $fg$ with expected running time $\bigot{s\log D}$, where $s,D$ are
    upper bounds on the total bit-sizes and degree of the inputs and outputs. Note that $s$ needs not be known in
    advance.
  \end{itemize}
  
  Unbalanced interpolation works by recovering terms in \emph{slices} according to coefficient size, starting with the
  largest ones (\cref{sec:interp}). It relies on a new technique (\cref{sec:slices}) that first recovers a superset of
  the exponents of terms of interest, in order to avoid carry propagation from smaller terms which have not yet been
  identified. We leverage the interpolation to handle unbalanced multiplication, exploiting a probabilistic product
  verification to discover adaptively the total bit-length of the product (\cref{sec:mul}).

  Our ultimate goal is to achieve quasi-optimal $\bigot{s}$ running time in terms of total bit-size of all coefficients
  and exponents; \cref{sec:supersparse} gives some insights on the barriers to achieving this so far. Along the way, we
  also carefully define a new black box model which acts as a precise specification for the interface between our
  various subroutines (\cref{sec:prelim}), and give some separation bounds on carry propagation between terms of
  different coefficient sizes (\cref{sec:carries}).

\subsection{Related work}\label{sec:related}

\paragraph{Dense polynomial multiplication}
The product of two polynomials $f,g\in \mathsf{R}[z]$
over a generic ring
can be computed using the schoolbook
algorithm, Karatsuba's algorithm~\cite{karatsuba1962}, Toom-Cook algorithm~\cite{Toom1963,Cook1966} or FFT-based
algorithms \emph{à la}
Sch\"onhage-Strassen~\cite{SchonhageStrassen1971,Nussbaumer1980,cantor1991}.
The best ring complexity is $O(D\log D\llog D)$ via the Cantor-Kaltofen algorithm~\cite{cantor1991}.
For finite fields one can get a better bit complexity using more specialized algorithms
~\cite{HHL17,HarveyHoeven2022}.

When $R =\zz$, one must take into account coefficient growth during the computation. In particular, FFT-based
algorithms require a height bound on the output. 
If $H$ is the height (i.e. largest coefficient magnitude) of the input polynomials, then
the height of the product is at most $H^2 D$, which is
$\bigo{\log(DH)}$ in terms of bit-length.

Then the fastest method for integer polynomial multiplication is
Kronecker substitution~\cite{Kronecker1882,Schonhage1982,Fateman2010,ParallelPolyIntMul2016},
which translates the problem into the multiplication of
two integers of bit-length $\gO(D\log(DH))$, by evaluating the
polynomials at sufficiently high power of 2. The complexity is therefore
$\bigo{D \log(DH) \log(D \log(DH))}$ using the fast integer multiplication~\cite{HarveyHoeven2021}.

None of these can efficiently handle unbalanced coefficients.
To the best of our knowledge, only the Toom-Cook
algorithm has been adapted to the case of unbalanced coefficients, by \citet{BodratoZanoni2020}. They reduce the problem
to bivariate Toom-Cook multiplication, but do not provide a formal
complexity analysis.

\paragraph{Sparse polynomial multiplication}
When $f$, $g$ are $t$-sparse polynomials as in \eqref{eqn:ff},
the schoolbook
algorithm for computing $fg$ requires $\gO(t^2)$ ring operations plus $\gO(t^2\log D)$ bit operations on the
exponents. In the worst case this is the best complexity possible,
though many practical improvements have been proposed
\cite{Johnson74,MonaganPearce2009,MonaganPearce2011}.
Output-sensitive algorithms have been designed to provide a better complexity when
fewer terms are expected \cite{Roche2011,vdHLec2012,vdHLebSch2013,ColeHariharan,ArRo15}.
More recent algorithms, based on sparse interpolation, managed
to reach quasi-linear dependency in the output sparsity
\cite{nakos2020,giorgi2020:sparsemul,GGPR22a}.

\paragraph{Sparse interpolation}
Modern algorithms to interpolate a sparse integer polynomial given a black box function for its evaluation started with
\citet{Zippel79,BenOrTiwari}, the latter of which is based on Prony's method from exponential analysis \cite{Prony}.
Numerous extensions have been proposed~\cite{Zippel1990,KY88,HuangGao2019} to handle finite field coefficients
\cite{GrigorievKarpinskiSinger1990,HuangRao1999,JavadiMonagan2010,GiRo2011,Huang2021}, discover the sparsity adaptively
via early termination \cite{KaltofenLee2003}, or extend to (sparse) rational function recovery
\cite{KaltofenTrager:STOC88,KYW1990,KaltofenYang2007,KaltofenNehring2011,CuytLee2011,HoevenLecerf2021}. Some algorithms
require slight extensions of the normal black box model
\cite{GrigorievKarpinskiSinger1990,Mansour1995,AlonMansour1995,MuraoFujise1996,KaltofenNehring2011,GiRo2011,BlaserJindal2014,vdHLec2014}.

Algorithms whose complexity is polynomial in $\log D$ rather than $D$
are termed \emph{supersparse}, and date from
\cite{Kal10a,Mansour1995,AlonMansour1995}.
\citet{GaSch09} gave a supersparse algorithm for \emph{white-box}
interpolation of a straight-line program; that technique has now been
extended to (some kinds of) black boxes over finite fields, or modular
rings \cite{ArGiRo14,ArGiRo16,HuangGao2020,Huang2023,Huang2019}.
More details on these algorithms and techniques can be found in
\cite{Arnold2016,vdHLec2019,armelle-phd}.
The first quasi-linear algorithm in terms of $\bigot{t\log(DH)}$
was recently given by the authors \cite{GGPR22a}.

Despite this considerable progress, no prior work provides a complexity estimate
in terms of the total (possibly unbalanced) bit-length of the output.

  \paragraph{Multivariate polynomials}
  A common approach to handling multivariate polynomials is to reduce to
  the univariate case via variable substitution
  \cite{Kal10a,ArRo14,HuangGao2019,GGPR22a,vdHLec2023Misc},
  where such polynomials naturally become (super)sparse.
  The univariate algorithms we propose will work well when the
  bit-length of exponents all have bit-lengths close to the maximum
  $\log D$, where the normal Kronecker substitution preserves the
  total exponent bit-length. But effectively handling unbalanced
  exponents remains an open problem.

\paragraph{Linear algebra}
Tangentially related to our work is the study of integer
and polynomial matrices with unbalanced entry sizes, where recent work
has focused on complexity in terms of average degree or total bit-length
rather than maximum.
A cornerstone result is for order basis computation
\cite{ZhoLab12}, which led to algorithms with similar complexity for
Hermite normal form, determinant, and rank, and recently Smith form of
an integer matrix
\cite{LabNeiZhou2017,LabNeiVuZhou2022,BirLabSto2023}.

\section{Preliminary}\label{sec:prelim}

\subsection{Bit lengths}

We use the usual notion of the bit-length of an integer, that is $ \bitlenp{a \in \nn}= \ceil{\log_2 (a + 1)}$
and $\bitlen{a \in \zz}= 1 + \bitlenp{\abs{a}}$.

The bit-length of a sparse polynomial $f$ as in \eqref{eqn:ff}
is defined as
\[\textstyle \bitlenx{f} = \sum_{i=1}^t \left(\bitlen{c_i} + \bitlenp{e_i} \right),\]
and as usual we define $\height{f} = \max_i \abs{c_i}$.

Useful for us will be this simple lower bound on the sparsity of
$f$ in terms of its bit-length.

\begin{lem}\label{lem:bitlen-lb}
  Let $f\in\zz[x]$ be a nonzero polynomial with bit-length $s = \bitlenx{f}$.
  The number of nonzero terms in $f$ is bounded by
  \(\#f < 2s/\log_2 s\).
\end{lem}
\begin{proof}
  For any $t\ge 0$, the smallest (in terms of bit-length) polynomial
  with $t$ terms is
  $1 + x + \cdots + x^{t-1}$,
  which has total bit-length
  \[\sum_{i=0}^{t-1} (\bitlen{1} + \bitlenp{i})
    = 2t + \sum_{i=1}^t \ceil{\log_2 i}
    > 2t + \log_2 t!
  \]
  Taking this value as $s$ and $t$ as $\#f$,
  the lemma is easily confirmed for all $s \le 16$.
  So assume $s \ge 17$.
  We first apply Stirling's approximation of factorial to
  get
  \[s = 2t + \log_2 t! > 2t + t\log_2\tfrac{t}{e}
    = t\log_2\tfrac{4t}{e} > t\log_2 t.\]

  Note also that because $s\ge 17$, we have
  $\log_2\log_2 s < \tfrac{1}{2}\log_2 s$.

  Finally, by way of contradiction, if $t \ge 2s/\log_2 s$, then
  \[t\log_2 t  > \tfrac{2s}{\log_2 s}(\log_2 s - \log_2 \log_2 s) > s,\]
  which is a contradiction to the inequality shown above.
\end{proof}

\subsection{Reducing bit length of a polynomial}

We will often need to reduce a sparse polynomial $f\in\zz[x]$
modulo $x^p-1$ over $\zz/m\zz$, for some integers $p,m$ with $p\le m$,
which we denote as $f \bmod\langle x^p-1, m\rangle$.
\begin{lem}\label{lem:reduce}
  Given $f\in\zz[x]$ and $p,m\in\zz$ with $p\le m$,
  computing  $f\bmod\langle
  x^p-1,m\rangle$ requires $\gO(\bitlenx{f}\llog m)$ bit
  operations.
\end{lem}
\begin{proof}
  Consider a single coefficient $c_i$. Reducing $c_i$ modulo $m$ requires
  $\gO\left(\frac{\log |c_i|}{\log m}\right)$ multiplications of
  $\lb(m)$-bit integers using the standard method, which
  has bit complexity $\gO(\log |c_i|\llog m)$. Summing over all
  coefficients, and doing the same for the exponents modulo $p$, gives the
  stated total bit complexity.
\end{proof}

\subsection{Black Boxes: MBBs and MDBBs}\label{ssec:mdbb}

We begin with the standard definition of a modular black box for evaluating a
sparse integer polynomial, with additional parameters $B,L$ for the
later complexity analysis.

\begin{definition}\label{def:mbb}
  A \emph{modular black box} (MBB) $\pi$ for unknown $f\in\zz[x]$
  parameterized by $B,L$ is
  a procedure which, given $a,m\in\nn$ where $a<m$, computes and
  returns $\pi(a,m)$ that equals
  $f(a) \bmod m$, using $\bigo{B + L\log m \llog m}$ bit
  operations.
\end{definition}

Intuitively, this corresponds to the notion that the MBB
works by performing $O(L)$ operations in $\zz/m\zz$ plus $O(B)$ other
operations. Moreover, in the final complexity measure, $B$ tracks the
total number of calls to the MBB, whereas $L$ (ignoring sub-logarithmic
factors) tracks the total size of the outputs produced by the MBB.

Our algorithms for interpolation and multiplication are composed of
multiple subroutines and interconnected procedures, both in this paper
and from prior work. Of particular importance is defining carefully the
interface of these procedures as it concerns the unknown polynomial
black box.

While the MBB model above is an appropriate starting point (and indeed
can serve as the input for our algorithms), in subroutines we frequently
need to update the unknown $f$ with an explicitly-constructed
partial result $\fs$. This update creates a problem for the MBB model, as
we do not know any quasi-linear time algorithm for MBB evaluation of a
known sparse polynomial where the degree and evaluation modulus may both
be large.
For instance, computing $a^b\bmod p$ when $\log b \approx \log p$ costs
$\gOt(\log^2 p)$ bit operations.

Instead, we will rely on a more specific black box for evaluation.

\begin{definition}\label{def:mdbb}
  A \emph{multi-point modular derivative black box} (MDBB) $\pi$
  for an unknown $f\in\zz[x]$ is
  a procedure which, given $p,\omega,m,k\in\nn$
  where $\omega$ is a $p$th primitive root of unity (PRU) in $\zz/m\zz$,
  produces two length-$k$ sequences containing the evaluations of
  both $f$ and $xf'$ at $1, \omega, \omega^2, \ldots, \omega^{k-1}$,
  where $f'$ is the formal derivative of $f$.
\end{definition}

We now show four efficient MDBB constructions: from a MBB, an given polynomial, or the sum or product of two MDBB.

Producing MDBB evaluations from a MBB black box follows the same
technique already present in prior work, which exploits the fact that
$(1+m)^e \equiv 1 + em \bmod m^2$, see e.g. \cite{ArRo15,GGPR22a}.

\begin{algorithm}\caption{$\MBBtoMDBB(\pi,p,\omega,m,k)$}\label{algo:MBBtoMDBB}
\begin{algorithmic}[1]
  \Input{MBB $\pi$ for unknown $f\in\zz[x]$ and
    MDBB inputs $p,\omega,m,k$}
  \Output{Evaluations of $f$ and $xf'$ at $1,\omega,\ldots,\omega^{k-1}$ modulo $m$}
  \For{$i = 0, 1, \ldots, k-1$}
    \State $\alpha_i \gets \pi(\omega^i, m^2)$
      \Comment{$\alpha_i = f(\omega^i)\bmod m^2$}
    \State $\beta_i \gets \pi((1+m)\omega^i, m^2)$
      \Comment{$\beta_i = f((1+m)\omega^i)\bmod m^2$}
    \State $\gamma_i \gets (\beta_i-\alpha_i)/m$ using exact integer division
  \EndFor
  \State\Return $(\alpha_i\bmod m)_{0 \le i < k}$ and $(\gamma_i)_{0\le i < k}$
\end{algorithmic}\end{algorithm}

\begin{lem}\label{lem:MBBtoMDBB}
  If $\pi$ is an MBB for $f\in\zz[x]$
  with cost parameters $B,L$ as in \cref{def:mbb},
  then \cref{algo:MBBtoMDBB} \MBBtoMDBB{} correctly produces a single
  set of MDBB evaluations of $f$ and $xf'$ and has bit complexity
  \(\bigo{Bk + (L+1)k\log m \llog m}.\)
\end{lem}
\begin{proof}
  Any single term
  $cx^e$ in $f$ is mapped to
  $(c+cem)x^e$ in $f((1+m)x)$ modulo $m^2$. Therefore
  $f((1+m)x) - f(x) \equiv xmf'(x)$, and the $\gamma_i$'s are indeed
  evaluations of $xf'(x)$ as required.
  The bit complexity comes from the cost of performing
  $O(k)$ MBB evaluations and ring operations in $\zz/m\zz$.
\end{proof}

To perform MDBB evaluations of a given polynomial $f$,
we compute the derivative $f'$ explicitly, reduce
both $f,f'$ modulo $\langle x^p-1,m \rangle$,
and then perform two multi-point evaluations.
If $p < \#f$, then the reduced polynomial is dense and we
use a DFT for the multi-point evaluations.
Otherwise, we treat it as a sparse multi-point evaluation
on points in geometric progression, which is a matrix-vector
product with a transposed Vandermonde matrix
(see e.g., \cite[Fact 3.2]{GGPR22a}).
The bit complexity is $\gOt(\min(p \log m + s \llog m, (s + k) \log p \log m)$,
or more precisely:

\begin{lem}\label{lem:extoMDBB}
  There exists a procedure \extoMDBB{} that,
  given a polynomial $f$ with $\bitlenx{f} = s$,
  and MDBB inputs $p,\omega,m,k\in\nn$,
  correctly computes the MDBB evaluations of $f$ and $f'$ in time
  \begin{align*}
    \gO\big(\min\big(
    &p\log p\log m\llog m + s\llog m,\\
    &(s + k\log s)\log(sp)\llog s\log m\llog m\big)\big).
  \end{align*}
\end{lem}

\begin{proof}
First we reduce $f$ modulo $\langle x^p-1, m\rangle$ in $\bigo{s\llog m}$
according to \cref{lem:reduce}, and in the same time
separately reduce exponents modulo $p$ and compute each coefficient-exponent product
$(ce\bmod m)$ to derive the coefficients of $xf'$.

Treating the reduced polynomials as dense with size $p$,
we can use Bluestein's algorithm for the DFT \cite{Blu70} in
$\bigo{p\log p}$ operations modulo $m$ to get the $k\leq p$ evaluations, i.e. $\omega\in\zz/m\zz$ is a $p$-PRU. 

Alternatively, we may use a sparse approach:
compute $\omega^{e_i}$ for each exponent in the reduced $f$
(which are the same as in $xf'$) in $\bigo{\#f \log p}$
and then perform two
$k\times \#f$ transposed Vandermonde matrix-vector products,
for a total of
$\bigo{(\#f+k)\log^2 \#f \llog \#f}$ operations modulo $m$.
Applying the bound on $\#f$ from \cref{lem:bitlen-lb} gives
the stated complexity.
\end{proof}

Finally we consider computing an MDBB for the sum or product
of two MDBB polynomials. The sum is straightforward; one can simply sum
the corresponding evaluations. For the evaluations of the derivative of
the product $h=fg$, we can use the product rule from elementary calculus:
\[xh'(x) = xf'(x) \cdot g(x) + xg'(x) \cdot f(x).\]

\begin{lem}\label{lem:sumMDBB}
  There exist procedures \sumMDBB and \prodMDBB which, given two MDBBs
  $\pi_1,\pi_2$ for unknown $f_1,f_2\in\zz[x]$, 
  and any MDBB input tuple $(p,\omega,m,k)$,
  compute respectively MDBB evaluations of the sum $f_1 + f_2$
  or product $f_1 f_2$.
  The cost is a single MDBB evaluation each of $\pi_1$ and $\pi_2$ plus
  (resp.)~
  $\gO(k\log m)$ or~$\gO(k\log m\llog m)$ bit operations.
\end{lem}

\subsection{Images of MDBBs}

Our main subroutines for unbalanced interpolation need
at each step to recover $f$ and $xf'$ modulo
$\langle x^p-1, m\rangle$, and to update such images
with respect to a partially-recovered explicit polynomial $\fs$.
The next (almost trivial) algorithm shows how to compute such modular images
from MDBB evaluations.

\begin{algorithm}\caption{$\MBBimg(\pi,p,\omega,m)$}\label{algo:MBBimg}
\begin{algorithmic}[1]
  \Input{MDBB $\pi$ for unknown $f\in\zz[x]$;
    and $p,\omega,m\in\nn$ s.t.\ $\omega$ is a $p$-PRU modulo $m$}
  \Output{$f$ and $xf'$ modulo $\langle x^p-1, m\rangle$}
  \State $(\alpha_i)_{0\le i < k}, (\gamma_i)_{0\le i < k}
    \gets \pi(p,\omega,m,p)$
  \State $\hat{\alpha}_0,\ldots,\hat{\alpha}_{p-1} \gets
    \IDFT_\omega(\alpha_0,\ldots,\alpha_{p-1})$
    over $\zz/m\zz$
  \State $\hat{\gamma}_0,\ldots,\hat{\gamma}_{p-1} \gets
    \IDFT_\omega(\gamma_0,\ldots,\gamma_{p-1})$
    over $\zz/m\zz$
  \State \Return
    $\sum_{i=0}^{t-1}\hat{\alpha}_i x^i$
    and
    $\sum_{i=0}^{t-1}\hat{\gamma}_i x^i$
\end{algorithmic}
\end{algorithm}

The running time is straightforward using Bluestein's algorithm \cite{Blu70}
for the IDFTs.

\begin{lem}\label{lem:MBBimg}
  \Cref{algo:MBBimg} \MBBimg{} always produces the correct output using 2 calls to the MDBB with $k=p$ and an additional
  \(\bigo{p\log m\log p \llog m}\) bit operations.
\end{lem}

\section{Coefficient collisions and carries}\label{sec:carries}

Our general approach to find the terms of $f$ with largest coefficients
will be to exploit the fact that there cannot be too many of them.
So we will take images of $f \bmod (x^p-1)$, for small values of $p$,
and extract the largest terms.
The goal of this %
section is to define certain
coefficient size boundaries, bound the probability of larger terms
colliding modulo $x^p-1$, and prove that the carries from smaller terms
colliding will not affect the larger terms too much.

For the remainder, assume $s,D,H$ are bounds (resp.)
on the bit-length, degree, and height of an unknown nonzero
$f\in\zz[x]$.

We define 4 categories of coefficients according
to size:

\begin{definition}\label{def:coeff-sizes}
  For any nonzero term $c x^e$ of $f$:
  \begin{itemize}
    \item If $\abs{c} < H^{1/6}$, we say $c$ is \emph{small}
    \item If $\abs{c} \ge H^{1/6}$, we say $c$ is \emph{medium}
    \item If $\abs{c} \ge \tfrac{1}{2}H^{13/30}$, we say $c$ is \emph{large}
    \item If $\abs{c} \ge H^{1/2}$, we say $c$ is \emph{huge}
  \end{itemize}
\end{definition}

Note that every huge term is also considered large, and every huge or
large term is also considered medium. The number of terms in each
category is also limited by $s$; for example, the number of huge terms
is at most $2s/\log_2 H$.

We now provide conditions on a randomly chosen prime $p$ to avoid
collisions modulo $(x^p-1)$. Formally, we say two nonzero terms $c_1 x^{e_1}$
and $c_2 x^{e_2}$ \emph{collide} modulo $p$ iff
$e_1 \equiv e_2 \bmod p$. The following lemma follows the standard
pattern from numerous previous results
(see for example \cite[Proposition 2.5]{GiorgiGrenetPerretduCray2023})
relying on the density of primes in an interval \cite{rosser1962}.

  \begin{lem}\label{lem:one-collision}
  Let $D,n>0$ be given and
  $e_0,e_1,\ldots,e_n$ be $n+1$ distinct exponents with each
  $0\le e_i \le D$.
  If $\lambda \ge \max\left(21, 3n\lb D\right)$, and
  if $p$ a random prime from
  $(\lambda,2\lambda)$, then with probability at least $\tfrac{1}{2}$,
  exponent $e_0$ does not collide with any of $e_1,\ldots,e_n$
  modulo $p$.
\end{lem}

By repeatedly choosing primes, we can ensure that every term in $f$
avoids collisions at least once. Note here that the set $a_1,\ldots,a_n$
need not be actual exponents in $f$.
\begin{cor}\label{cor:many-collisions}
  Let $D,n>0$ be given and $a_1,\ldots,a_n$ be distinct integers
  all satisfying $0\le a_i \le D$, and let $P$ be a list of primes
  in $(\lambda,2\lambda)$, each chosen independently and uniformly at
  random.
  If $\lambda \ge \max(21,3n\lb D)$ and $\#P\ge 2\lb s + 3$, then with
  probability at least $1-1/(4s\lb s)$, for every term
  $c_ix^{e_i}$ of $f$, there exists at least one $p\in P$ such that
  $e_i$ does not collide (modulo $p$) with any $a_j\ne e_i$.
\end{cor}
\begin{proof}
  For any single term, by \cref{lem:one-collision} the probability it
  is in collision with at least one of the $a_j$'s
  for \emph{every} $p\in P$ is
  at most $1/2^{2\lb s + 3} < 1/(8s^2)$.
  Applying the bound on $\#f$ from \cref{lem:bitlen-lb}, and taking the
  union bound, gives the stated result.
\end{proof}

Similar to prior works on sparse interpolation, our algorithm will
recover the exponent $e_0$ of a term $c_0 x^{e_0}$ by division of
coefficients between $f$ and $xf'$ as produced by the MDBB. But this
division is not exactly equal to $e_0 c_0 / c_0$ because we cannot avoid
collisions with \emph{small} terms in $f$. The following crucial lemma
shows that such collisions still allow exact recovery of the exponent
$e_0$ whenever the coefficient $c_0$ is large.

\begin{lem}\label{lem:errors}
  Suppose $\lb H \ge \max\left(61,\ 15\lb s,\ 6\lb D\right)$ and
  $p,m\in\nn$ with $m\ge 4H^{7/6}$,
  and let $c_0x^{e_0}$ be any large term of $f$ as in \cref{def:coeff-sizes}.
  If this term does not collide with any other medium term modulo $p$,
  then $e_0$ can be accurately recovered by
  approximate division of coefficients in $f$ and $f'$ modulo
  $\langle x^p-1, m\rangle$.
\end{lem}
\begin{proof}
  Let $\mathcal{S}$ be the set of indices of terms $c_ix^{e_i}$ in $f$
  that collide with $e_0$, so that $e_i\equiv e_0 \bmod p$ for all
  $i\in\mathcal{S}$. By assumption, all corresponding coefficients are
  ``small'' as in \cref{def:coeff-sizes}, so each
  each $\abs{c_i} < H^{1/6}$ for $i \in \mathcal{S}$.

  First we establish a bound on the magnitude of the sum of any
  number of small coefficients.
  Using the fact that
  $\#\mathcal{S} \le \#f \le s$ and the assumption $\lb H \ge 15\lb s$, we have,
  for any set of small term indices $\mathcal{S}$,
  \begin{equation}\label{eqn:small-bound}
    \textstyle\abs{\sum_{i\in\mathcal{S}} c_i}
    \le \sum_{i\in\mathcal{S}}\abs{c_i}
    < \#\mathcal{S} H^{1/6} \le s H^{1/6} \le H^{7/30}.
  \end{equation}

  Then the terms corresponding to $c_0x^{e_0}$ in $xf'$ and in $f$
  mod $\langle x^p-1, m\rangle$ are respectively
  \begin{eqnarray*}
    \textstyle c_0e_0+\sum_{i\in\mathcal{S}}c_ie_i &\mod m \\
    \textstyle c_0+\sum_{i\in\mathcal{S}} c_i &\mod m
  \end{eqnarray*}

  Note that, from the bound \eqref{eqn:small-bound} above, and the facts
  that $\abs{c_0}\le H$ and $e_i\le D\le H^{1/6}$, both of these
  are less than $2H^{7/6}$ in absolute value, and therefore modulo
  $m \ge 4H^{7/6}$ they can be recovered exactly as signed integers.

  Our aim is to show that the quotient of these integers, rounded to
  the nearest integer, equals $e_0$ exactly, or equivalently that
  \begin{equation}\label{eqn:errors}
    \abs{\frac{\textstyle\sum_{i\in\mathcal{S}} c_ie_i - e_0\sum_{i\in\mathcal{S}}c_i}%
      {\textstyle c_0 + \sum_{i\in\mathcal{S}}c_i}}
    < \frac{1}{2}
  \end{equation}

  Using \eqref{eqn:small-bound}, along with the bound $D\le H^{1/6}$ from the
  statement of the lemma, the numerator magnitude from
  \eqref{eqn:errors} is at most
  \[\textstyle
    \abs{\sum_{i\in\mathcal{S}} c_i(e_i-e_0)}
    \le D\sum_{i\in\mathcal{S}}\abs{c_i}
    < DH^{7/30}
    \le H^{2/5}.
  \]

  A lower bound on the denominator from \eqref{eqn:errors}
  can be obtained similarly, since $\abs{c_0}\ge \frac{1}{2}H^{13/30}$ by the
  definition of a large term:
  \[\textstyle
    \abs{c_0 + \sum_{i\in\mathcal{S}}c_i}
    \ge \abs{c_0} - \sum_{i\in\mathcal{S}}\abs{c_i}
    \ge H^{2/5}(\tfrac{1}{2}H^{1/30} - H^{-5/30})\]
  Because $\lb H \ge 61$, one can easily confirm that
  $\tfrac{1}{2}H^{1/30} - H^{-5/30} > 2$.
  Therefore the denominator is at least $2H^{2/5}$ and the inequality from
  \eqref{eqn:errors} is satisfied.
\end{proof}

\section{Recovering huge terms}\label{sec:slices}

This section presents the heart of our new algorithm for
unbalanced interpolation,
\cref{algo:UintSlice} \UintSlice, which recovers an explicit
partial interpolant $f^*\in\zz[x]$ such that, with high probability,
$\height{f-f^*}\le \sqrt{H}$.

Recall the formal notion of small/medium/large/huge terms from
\cref{def:coeff-sizes}.
Slice interpolation works in two phases, first calling subroutine
\cref{algo:SupSup} \SupSup{} to find (a superset of) the support of
all large terms, and secondly using this set to reliably recover the
huge terms' coefficients.

{ All the following algorithms need to find a $p$-PRU in $\zz/m\zz$ for a random prime $p$ in $(\lambda,2\lambda)$,
  and a sufficiently large integer $m$. For that, we rely on \cite[Fact 2.2]{GGPR22a} to find \emph{w.h.p} a triple
  $(p,q,\omega)$ such that $\omega$ is a $p$-PRU in $\fq$, and use \cite[Algorithm LiftPRU]{{GGPR22a}} to lift $\omega$
  to a $p$-PRU modulo $m$, where $m=q^k$ is large enough. This will require $\bigot{\log H \log \lambda} +
  \polylog(\lambda)$ bit operations for $\log m\in\bigo{\log H}$. Since $\log H = \bigo{s}$ this cost never dominates
  in our  analysis.}

The first subroutine, \cref{algo:SupSup} \SupSup{},
works by sampling from \MBBimg{}
multiple times, for sufficiently large
exponent moduli $p$ so that most of the medium (or larger) terms do not
collide. From \cref{lem:errors} in the previous section, even with some
small-term collisions, we will be able to accurately recover the
exponents of any large terms. Erroneous
entries in \esuper{} will likely result from medium/large/huge collisions, but
that is acceptable as long as we don't ``miss'' any true exponents of
large terms.

\begin{algorithm}\caption{$\SupSup(\pi,s,D,H)$}\label{algo:SupSup}
\begin{algorithmic}[1]
  \Input{MDBB $\pi$ for unknown $f\in\zz[x]$,
    bounds $s,D,H$ on (resp.)
    the bit-length, degree, and height of $f$,
    satisfying the conditions of \cref{lem:errors}}
  \Output{Set $\esuper\subset \nn$ which contains the exponents of all
    terms in $f$ with large (or huge) coefficients w.h.p.}
  \State $\lambda \gets \max(21, 18 s \lb D / \lb H)$;\quad
    $\esuper \gets \{\}$
  \For{$i = 1,2,\ldots,\ceil{2\lb s}+3$}
  \State Choose random prime $p\in(\lambda,2\lambda)$
  \State Construct $m,\omega$ with $m\ge 4H^{7/6}$,
      and $\omega$ a $p$-PRU mod $m$
  \State $g, h \gets \MBBimg(\pi,p,\omega,m)$

      \ForEach{corresponding terms $ax^e, bx^e$ in $g$ and $h$}
      \If{$\abs{a} \ge \frac{1}{2}H^{13/30}$ and $0 \le \round{b/a} \le D$}
        \If{$\#\esuper < 60s(\lb s + 2)/(13\lb H)$}
          \State $\esuper\gets\esuper\cup\{\round{b/a}\}$
        \Else \ \Return $\{\}$
        \EndIf
      \EndIf
    \EndFor
  \EndFor
  \State \Return \esuper
\end{algorithmic}\end{algorithm}

\begin{lem}\label{lem:SupSup}
  \Cref{algo:SupSup} \SupSup{} always produces a set
  $\esuper\subseteq\{0,\ldots,D\}$ with
  $\#\esuper \in \bigo{s\log s / \log H}$;
  makes $\bigo{\log s}$ MDBB calls with
    $p,k\in O(s\log D / \log H)$ and
    $\log m\in O(\log H)$;
  and uses $\bigo{s\log D \log^2 s \llog H)}$
  additional bit operations.
  With probability at least $1-1/(4s\lb s)$, \esuper{}
  includes the exponents of every large term of $f$.
\end{lem}
\begin{proof}
  Let $c x^e$ be an arbitrary large term of $f$.
  By \cref{def:coeff-sizes},
  the number of medium terms in $f$ is at most $6s/\lb H$.
  Applying \cref{cor:many-collisions}, we see that with at least
  the stated probability, for each large term of $f$ there exists
  an iteration $i$ where that term does not collide with any other
  medium (or larger) term.
  And then by \cref{lem:errors}, the actual exponent of that
  term is accurately recovered from rounding division of coefficients
  on step $i$ and added to \esuper{}.
  This proves the stated probabilistic correctness.

  For the size of \esuper{},
  observe that
  the bit-length of
  $f\bmod(x^p-1)$ is at most $s$, so there are at most
  $\tfrac{30}{13}s/\lb H$ new integers added to \esuper{} in each of the
  $\bigo{\log s}$ iterations.

  The bit complexity is dominated by the cost of the
  $\bigo{\log s}$ calls to \MBBimg{}, which by the choices of $p,m$ and
  from \cref{lem:MBBimg} gives the stated cost.
\end{proof}

We now use the set \esuper{} to recover the coefficients of the huge terms (only). The same number of MDBBImage
evaluations are performed, but with slightly larger p values to avoid collisions with all exponents (whether correct or
not) in \esuper{}. Because \esuper{} contains the exponents of all large terms, and any errors from medium and small terms
cannot propagate up to the ``huge'' level, every huge term obtained on this step is accurately added to the result.

\begin{algorithm}\caption{$\UintSlice(\pi,s,D,H)$}\label{algo:UintSlice}
\begin{algorithmic}[1]
  \Input{MDBB $\pi$ for unknown $f\in\zz[x]$,
    and bounds $s,D,H$ on (resp.)
    the bit-length, degree, and height of $f$,
    satisfying the conditions of \cref{lem:errors}}
  \Output{$\fs\in\zz[x]$ such that
    w.h.p.\ $\height{f-\fs} \le \sqrt{H}$}
  \State $\esuper \gets \SupSup(\pi,s,D,H)$
  \State $\lambda \gets \max(21, 3\cdot \#\esuper\cdot \lb D)$;\quad
    $\fs \gets 0$
  \For{$i = 1,2,\ldots,\ceil{2\lb s}+3$}
    \State Choose random prime $p\in(\lambda,2\lambda)$
    \State Construct $m,\omega$ with $m\ge 2H$
    and $\omega$ a $p$-PRU mod $m$
    \State $\pi^* \gets \sumMDBB(\pi,\extoMDBB(-\fs))$
    \State $g.h \gets \MBBimg(\pi^*,p,\omega,m)$%
    \State $E_j\gets\{\}$ for $0\le j < p$
    \ForEach{$a \in \esuper$}
      \State $E_{a\brem p} \gets E_{a\brem p} \cup \{a\}$
    \EndFor
    \ForEach{term $cx^i$ in $g$}
      \If{$\abs{c} \ge \frac{1}{2}H^{1/2}$ and $\#E_i = 1$}
        \State $e \gets$ the unique element of $E_i$
        \If{$\fs$ does not contain a term with $x^e$}
          \State $\fs \gets \fs + cx^e$
          \If{$\bitlenx{\fs} > s$}\ \Return 0
          \EndIf
        \EndIf
      \EndIf
    \EndFor
  \EndFor
  \State \Return $\fs$
\end{algorithmic}\end{algorithm}

\begin{lem}\label{lem:UintSlice}
  \Cref{algo:UintSlice} \UintSlice{}
  always produces a polynomial $\fs\in\zz[x]$ with
  bit-length at most $s$,
  makes $\bigo{\log s}$ MDBB calls with
  $p,k\in\bigo{s\log D\log s/\log H}$
  and $\log m \in \bigo{\log H}$,
  and uses $\bigo{s\log D\log^3 s\llog H}$
  additional bit operations.
  It returns $\fs\in\zz[x]$ such that, with probability at least
  $1 - 1/(2s\lb s)$, the height of $(f-\fs)$ is less than $\sqrt{H}$.
\end{lem}
\begin{proof}
  Consider any set of non-large terms. By the upper bound on the total number of terms from \cref{lem:bitlen-lb}, the
  condition that $\lb H \ge 15\lb s$ from \cref{lem:errors}, and from \cref{def:coeff-sizes}, the sum of these non-large
  terms is at most
  \begin{equation}\label{eqn:medsum}
    \sum\abs{c_i} < \tfrac{1}{2}sH^{13/30} \le \tfrac{1}{2}\sqrt{H}.
  \end{equation}

  For the probabilistic correctness, suppose that (i)
  the call to \cref{algo:SupSup} \SupSup{} correctly returns a set
  $\esuper$ which contains the exponents of all large terms, and
  (ii) each large or huge term does not collide with any of the
  other exponents in $\esuper$ for at least one of the chosen $p$'s in
  the outer for loop. Taking a union bound with \cref{lem:SupSup} and
  \cref{cor:many-collisions}, both (i) and (ii) are true with probability
  at least $1-1/(2s\lb s)$, as required. We now prove the
  algorithm returns $f^*$ correctly under these two assumptions.

  Consider any huge term $c_0 x^{e_0}$ in $f$, i.e., with
  $\abs{c_0}\ge \sqrt{H}$.
  From \eqref{eqn:medsum}, the resulting collision with any number
  of non-large terms will still result in a coefficient larger than
  $\tfrac{1}{2}\sqrt{H}$ in absolute value, so this term will be added
  to $f^*$.
  Conversely, consider any term $c x^e$ which is added to $\fs$ in the
  algorithm. By assumption that $\esuper$ is correct,
  $c$ must be the sum of exactly one large (or huge) term plus some
  number of non-large coefficients. By \eqref{eqn:medsum} again, the
  coefficient of this term in $f-f^*$ is at most the size of the colliding
  non-large coefficients' sum, which is less than $\frac{1}{2}\sqrt{H}$.

  Because $f^*$ contains some term corresponding to each huge term in
  $f$, and every term added to $f^*$ reduces the height of that
  coefficient in $f-f^*$ below $\sqrt{H}$, we conclude that
  $\height{f-f^*}<\sqrt{H}$.
  For the complexity analysis,
  we have $\#\esuper \in \bigo{s\log s/\log H}$ from \cref{lem:SupSup},
  which means $p\in\bigo{s\log D\log s/\log H}$.
  As before, the cost of the evaluations $\MBBimg$ dominates the
  complexity, which comes from \cref{lem:MBBimg}.
\end{proof}

In most cases, \UintSlice{} will be called with a MDBB $\pi$ which is a sum of an actual unknown black-box polynomial
minus an explicit partial interpolant $\fs\in\zz[x]$ recovered so far, via \extoMDBB{} and \sumMDBB{}. As soon as
$\bitlenx{\fs} \leq s$, next corollary shows that the bit complexity is not affected.

\begin{cor}\label{cor:UintSliceExp}
  Let $\pi$ be any MDBB and $\fs\in\zz[x]$ such that
  $\bitlenx{\fs}\le s$. Then
  calling $\UintSlice(\pi^*,s,D,H)$ with
  $\pi^* = \sumMDBB(\pi, \extoMDBB(-\fs))$
  has the same asymptotic cost as $\UintSlice(\pi,s,D,H)$.
\end{cor}
\begin{proof}
  From \cref{lem:extoMDBB} and \cref{lem:sumMDBB}, the cost of a single
  MDBB evaluation of $\pi^*$ with $p,\omega,m,k$ is a single call to $\pi$
  plus 
  \[\bigo{s\llog m + p\log p\log m\llog m}\]
  bit operations. When $k\in\bigo{p}$,
  $\log m\in\bigo{\log H}$, and
  $p\log m\in\bigo{s\log D\log s}$, this simplifies to
  \(\bigo{s\log D \log^2 s \llog H}\).

  Then from \cref{lem:UintSlice}, the total cost of calls to $\pi^*$ in \UintSlice{} is the same calls to $\pi$, plus an
  additional bit cost which is always bounded by the bit-cost of \UintSlice{} itself.
\end{proof}

\section{Unbalanced interpolation}\label{sec:interp}

Our overall algorithm calls \cref{algo:UintSlice} \UintSlice{} repeatedly starting with $H=s$, until the height
decreases such that the conditions of \cref{lem:errors} are no longer satisfied, and then switches to a \emph{balanced}
sparse interpolation algorithm for the base case. For that, we use Algorithm 2 \bcinterp{} from \cite{GGPR22a}, which
solves the problem in $\bigot{T\log D + T\log H}$ time. In the base case when \cref{lem:errors} no longer applies, we
have $\log H \in \bigo{\log D + \log s}$, so this complexity becomes simply $\bigot{s\log D}$ as required.

The only slight change from \cite{GGPR22a} is that $\bcinterp$ must take a MDBB rather than a usual MBB as the
black-box for unknown $f\in\zz[x]$. But the MDBB model already fits the case perfectly; the algorithm in \cite{GGPR22a}
is using the MBB only to evaluate $f$ at consecutive powers of a $p$-PRU, which is precisely what a MDBB does already.

All this discussion is summarized in the following claim, where the precise complexity bound comes from Theorem 3.18 in
\citep{armelle-phd}.

\begin{fact}\label{fact:bcinterp}
  Let $\pi$ be a MDBB for $f\in\zz[x]$, and bounds $D,T,H$ on (resp.) the degree, sparsity, and
  height of $f$. Write $n = T(\lb D + \lb H)$.
  Then $\bcinterp(\pi,D,T,H)$ is a Monte-Carlo algorithm that always produces a polynomial $\fs$ within the given bounds
  using $\bigo{\log T}$ calls to the MDBB and $\bigo{n\log^3 n \log^2 T (\llog n)^2}$ bit operations. 
  In each MDBB call, we have $p\in\bigo{n\log T}$ and $\log m \in \bigo{\log(DH)}$, and the sum of $k$ values over all
  $\bigo{\log T}$ calls is $\bigo{T}$. If the unknown $f$  satisfies the given bounds $D,T,H$, then with
  probability at least $\tfrac{2}{3}$, the returned polynomial $\fs$ equals $f$.
\end{fact}

As mentioned in \cite{GGPR22a}, one can decrease the failure probability to any $\epsilon>0$ by iterating the call at
least $48\ln\tfrac{1}{\epsilon}$ times and returning the majority result. We now present our main algorithm that
recovers an unknown polynomial with possibly unbalanced coefficients from a given MDBB and bounds only on the degree and
total bit-length.

\begin{algorithm}\caption{$\Uint(\pi,s,D)$}\label{algo:Uint}
\begin{algorithmic}[1]
  \Input{MDBB $\pi$ for unknown $f\in\zz[x]$,
    and bounds $s,D$ on (resp.)
    the bit-length and degree of $f$}
  \Output{$f^*\in\zz[x]$ such that $f=f^*$ w.h.p.}
  \State $H \gets 2^s$; \quad $\fs \gets 0$; \quad $\pi^* \gets \pi$
  \While{$H \ge \max\left(61,\ 15\lb s,\ 6\lb D\right)$}
    \State $f^* \gets \fs + \UintSlice(\pi^*, s, D, H)$
    \If{$\bitlenx{\fs} > s$}\ \Return 0 \EndIf
    \State $\pi^* \gets \sumMDBB(\pi,\extoMDBB(-\fs))$
    \State $H \gets \sqrt{H}$
  \EndWhile
  \State $R \gets$ empty list
  \For{$i = 1,2,\ldots, \lceil 48 \ln (2s) \rceil$}
    \State $\fss \gets \bcinterp(\pi^*, 2s/\lb s, D, H)$
    \State Append $\fss$ to $R$
  \EndFor
  \If{$R$ has a majority element $f^{**}$ and $\bitlenx{f^{**}}\le s$}
    \State\Return $\fss$
  \Else\ \Return $0$
  \EndIf
\end{algorithmic}\end{algorithm}

The precise cost estimate is given in the following theorem.

\begin{thm}\label{thm:Uint}
  \cref{algo:Uint} \Uint{} always returns a polynomial with bit-length
  at most $s$ and uses $\bigo{\log^2 s}$ MDBB calls plus an
  additional
  \[\bigo{s\log D \log^5 s (\llog s)^2}\]
  bit operations.
  Each MDBB call has $p\in \bigo{s\log D}$,
  $\log m \in \bigo{s}$, but $(p+k)\log m\in \bigo{s\log D \log s}$;
  furthermore, the sum of $k$ over all MDBB calls is $\bigo{s\log D\log s}$.
  If the unknown $f$ actually satisfies bit-length and degree bounds
  $s,D$, then with probability at least $1-1/s$, the returned $\fss$
  equals $f$.
\end{thm}

\begin{proof}
  First observe that due to the checks in the algorithm, $\fs$ and %
  $\fss$ always have bit-length at most
  $s$ and degree at most $D$.

  For the bit complexity, notice that at the end of the first loop we have $\log H \in \bigo{\log D + \log s}$. And then
  because the bit-length $s$ could never be more than $D\lb H$, we see that in fact $\log s + \log H \in \bigo{\log D}$
  at this point in the algorithm. Applying \cref{fact:bcinterp}, the total bit-cost of the $\bigo{\log s}$ calls to
  \bcinterp{} dominates the (non-black box) bit complexity and gives the stated bound.

  The black box calls are dominated instead by the first loop, whose cost is verified by applying \cref{lem:UintSlice}
  and observing that, since $\log H$ is exponentially decreasing from $s$ down to at least $\log s$, the sum of $k$
  values in all calls to \cref{lem:UintSlice} is as claimed.%

  For probabilistic correctness, the union bound over all calls to
  \UintSlice{}, \cref{lem:UintSlice}
  gives at most a $\tfrac{1}{2s}$ chance that \emph{any}
  of those steps does not return the correct result.
  In the subsequent calls to \bcinterp{},
  setting $T=2s/\lb s$ is valid by \cref{lem:bitlen-lb}.
  From the preceding discussion, we have the same probability of failure
  of the majority vote approach at the end. Therefore the total failure
  probability is at most $1/s$ as claimed.
\end{proof}

Similarly to \cref{cor:UintSliceExp}, we first clarify that adding
an explicit polynomial to the black box $\pi$ does not affect the
bit complexity. This is clearly the case as the MDBB $\pi^*$ in the algorithm
in fact already includes a sum with an explicit polynomial whose
bit-length is at most $s$.
\begin{cor}\label{cor:UintExp}
  Let $\pi$ be any MDBB and $\fs\in\zz[x]$ such that
  $\bitlenx{\fs}\le s$. Then
  calling $\Uint(\pi^*,s,D)$ with
  $\pi^* = \sumMDBB(\pi, \extoMDBB(-\fs))$
  has the same asymptotic cost as $\Uint(\pi,s,D)$.
\end{cor}

The following corollary gives the concrete cost to run \Uint{}
when the input is a modular black box (MBB). The stated bit complexity,
which provides the clearest comparison to prior works in the literature,
can be even more simply stated as $\bigot{(B+L)s\log D}$.

\begin{cor}\label{cor:uinterp-mbb}
  If $\pi$ is a MBB for unknown $f\in\zz[x]$
  with cost parameters $B,L$ as in \cref{def:mbb},
  and if $f$ has bit-length at most $s$ and degree at most $D$,
  then $\Uint(\MBBtoMDBB(\pi),s,D)$
  correctly return $f$ with probability at least $1-\tfrac{1}{s}$
  and  total bit cost 
  \[\bigo{\left(B + L\log^3 s + \log^4 s (\llog s)^2\right)s\log D\log s}\]
\end{cor}
\begin{proof}
  Follows directly from \cref{thm:Uint} and \cref{lem:MBBtoMDBB}.
\end{proof}

\section{Unbalanced multiplication}\label{sec:mul}
In this section we will show how to multiply two polynomials
$f,g\in\zz[x]$ with degree less than $D$ and bit-lengths $\ell$,
in time $\bigot{s\log D}$, where
$s=2\ell + \bitlenx{fg}$ is the total input and output bit-length.

As motivated in the introduction, this is the first algorithm which is sub-quadratic in $s$ for unbalanced coefficients,
\emph{even for dense polynomials}. For sparse integer polynomial multiplication, the state of the art comes from
\cite[Algorithm 20]{armelle-phd}, based on the (balanced) sparse interpolation of \cite{GGPR22a}, having a bit complexity
$\bigot{t\log H + t\log D}$, which is $\gOt(s^2+s \log D)$ in terms of the total bit-length $s$. For the remainder we
address only the case of sparse polynomials as the dense case works in exactly the same way.

Our general approach is to construct a MDBB for the product $fg$ and then use \cref{algo:Uint} \Uint{} to interpolate it
with complexity dependent on (unbalanced) total bit-length of the inputs and output. But we will show that the
pessimistic bound on $\bitlenx{fg}$ is quadratic in the input sizes, which means that a na\"ive use of \Uint{} would not
be quasi-linear in the actual size of $fg$.

Instead, we develop an efficient \emph{probabilistic} verification method,
which allows us to try interpolating $fg$ with smaller bounds on the
output bit-length, then repeatedly doubling the \emph{optimistic} bound
until the verification passes.

\paragraph{{Unbalanced product bit-length}}
Suppose $f,g\in\zz[x]$ have $t$ non-zero terms and total bit-length
$\ell$.
The number of terms in $h=fg$ is at most $t^2$ but the total bit-length of $h$ remains $\bigo{t\ell}$, instead of $\bigo{t^2 (\ell + \log t)}$.
Indeed, when no collision occurs, every entry of $f$ and $g$ will contribute exactly $t$ times to the bit length of $h$.
Collision among non-zero terms can only decrease the sum of the bit-length of the $t^2$ coefficient products. Hence the
bit-length of $h$ is bounded by $t$ times the bit-length of the input. Using Lemma~\ref{lem:bitlen-lb} we can derive
the following lemma.

\begin{lem} \label{cor:bitlenprod}
If $\bitlenx{f},\bitlenx{g}\le\ell$, then
$\bitlenx{fg}\le 4\ell^2/\log \ell$
and $\height{fg}\le 4^\ell\cdot{}\ell$
\end{lem}

\subsection{Unbalanced product verification}

Our unbalanced interpolation algorithm \Uint{} takes a bit-length bound
and always returns a polynomial with that size, but gives (obviously) no
correctness guarantee if that bound was too small.
To avoid relying on the pessimistic upper bound on $\bitlenx{fg}$, we
provide an efficient verification method.
While deterministic verification seems to be a difficult task,
a straightforward analysis of the probabilistic algorithm
proposed in \cite[Section 5.2]{GiorgiGrenetPerretduCray2023} is already satisfactory
to reach the needed complexity:

\begin{lem}\label{lem:verifdense}
  There is a one-sided error randomized algorithm, called \VerifProd($f,g,h,\epsilon$), that, given 
  $f$, $g$, $h\in\zz[x]$ with degrees and bit-lengths at most $D$ and $s$,
  and failure probability $\epsilon\in(0,1)$,
  tests whether $h = f\times
  g$ with false-positive probability at most $\epsilon$, using $\gO(s\log \frac{s}{\epsilon}\llog \frac{s}{\epsilon}
    + s\log \frac{D}{\epsilon} \llog \frac{D}{\epsilon})$ bit operations.
\end{lem}

\begin{proof}
  The approach is exactly the same as in \cite[Section 5.2]{GiorgiGrenetPerretduCray2023}.
  We choose random primes $p,q$, evaluate
  $\delta = (h-fg)\bmod\langle x^p-1, q\rangle$
  at a random point $\alpha\in\fq$, and check whether the result is
  zero.
  If $h=fg$, then $\delta=0$ and the check always succeeds.
  Otherwise, from
  \cref{cor:bitlenprod} and \cref{lem:reduce} one may show that taking $p=\Omega(\frac{1}{\epsilon}\max(s,D))$ and $q=
  \gO(\frac{1}{\epsilon}s)$ is enough to guarantee that $\delta$ is zero with a probability less than
  $\frac{2}{3}\epsilon$. Using \cite[Theorem 3.1]{GiorgiGrenetPerretduCray2023} one can show that evaluating
  $\delta(\alpha)=0$ costs the claimed complexity and that $\alpha$ is a root of $\delta$ with probability less than
  $\frac{1}{3}\epsilon$.
\end{proof}

\subsection{Adaptative polynomial multiplication}

One can easily construct an MDBB $\pi$ for the product of two polynomials $f,g\in \zz[x]$ by composing the procedures
\prodMDBB and \extoMDBB given in \cref{ssec:mdbb}. Feeding this MDBB into our unbalanced interpolation algorithm \Uint{}
with a correct bit-length bound yields a Monte-Carlo algorithm for computing the product $fg$. {We now provide a
  \emph{probably correct and probably fast} randomized algorithm (sometime called Atlantic-City) for that computation
  where the correct bit-length is discovered through our probabilistic verification of the computed result.}

\begin{algorithm}
  \caption{\UnbalancedProd}\label{algo:UnbalancedProd}
\begin{algorithmic}[1]
  \Input{$f,g \in\zz[x]$ of degree at most $D$ and bit-lengths $\ell$;}
  \Output{$h\in\zz[x]$ such that $h=f\times g$ w.h.p.}
  \State $h \gets 0$; $s \gets \ell$
  \State $s_{max}\gets 2\ell+ 4\ell^2/\log_2(\ell)$;
  $\epsilon \gets (s_{max} (8\ell+4))^{-1} $ 
  \State  $\pi \gets \prodMDBB(\extoMDBB(f), \extoMDBB(g))$
  \While{$s<2s_{max}$ and {\bf not} \VerifProd($f,g,h,\epsilon$)}
  \State $h \gets$ \Uint($\pi, s, 2D$)
  \State $s \gets 2s$
  \EndWhile
  \State\Return $h$
\end{algorithmic}
\end{algorithm}

\begin{thm}\label{thm:unbalancedproduct}
  Let two polynomials $f,g\in\zz[x]$ of degree at most $D$ and bit-lengths $\ell$. Algorithm \UnbalancedProd is an
  Atlantic-City algorithm that returns the polynomial $h=f\times g $ with probability at least $1-\frac{1}{s}$, using an
  expected total of $\gO(s\log D\log^5s(\log\log s)^2)$ bit operations where $s=\gO(\ell+\bitlenx{h})$. This is the
  actual complexity with probability at least $1-\frac{1}{s}$. In the worst-case it requires $\gO(s^2\log
  D\log^4s(\log\log s)^2)$ bit operations.
\end{thm}

\begin{proof}
  First one should note that according to \cref{cor:bitlenprod}, $s_{max}$ bounds
  $\bitlenx{fg}+\bitlenx{f}+\bitlenx{g}$. So taking $s$ at most $2s_{max}-1$ is enough to guarantee at least one
  iteration is done with a correct bit-length. This iteration will fail with the same probability as \Uint, which is
  $<1/s$. Let us assume that \Uinterpolate never produces a correct answer even when $s$ correctly bounds the bit-length
  of $fg$. Since the algorithm requires at most $\lceil \lb (2s_{max}/l) \rceil < 8\ell+4$ iterations before
  terminating, the probability that at least one verification test fails is $\le (8\ell+2)\epsilon = 1/s_{max}$. Hence,
  the probability of a wrong answer is $< 1 /s_{max} $.

  For the running time we note that all returned polynomials %
  have a bit-length smaller than $s$ even
  for wrong results. Therefore, each iteration of \Uint has a cost of $\gO(s\log D\log^5s(\log\log s)^2)$ bit
  operations, plus the MDBB evaluations.

  Evaluating the MDBB $\pi$ entails one \prodMDBB{} and two
  \extoMDBB{}s on polynomials whose bit-length is at most $\ell$.
  From \cref{lem:extoMDBB} and \cref{lem:sumMDBB} this is dominated by the cost
  of \extoMDBB{}, and then by \cref{cor:UintExp} the cost of all MDBB evaluations
  does not dominate the bit-cost of \Uint{} itself.

  The bit cost from \Uint{} dominates the complexity of the corresponding call
  to \VerifProd because
  $\log\frac{1}{\epsilon}=\bigo{\log \ell} = \gO(\log s)$.
  Since the value of $s$ is doubling at each iteration, the global running time is dominated by the cost
  of the last one, which is done with $s=\bigo{\bitlenx{fg}}$.

  It remains to see if there are \emph{extra} iterations after reaching a bound $s_h\geq \bitlenx{fg}$. There are
  $\gamma < 8\ell+4$ extra iterations if both \Uinterpolate and \VerifProd fail $\gamma$ times in a row. Even ignoring
  the probability of failure of \VerifProd, this cannot happen with probability more than $(\frac{1}{s_h})^\gamma$.
  Let $C(s)$ be the cost of one iteration. Since it is quasi-linear in $s$, the $i$-th extra iteration costs $\gO(2^i
  C(s_h))$. Hence the expected cost is $\gO\left(\sum_{i=1}^{\gamma} 2^{i}C(s_h)(\frac{1}{s_h})^i\right) =
  \gO(C(s_h))$.
\end{proof}

We remark that one might hope to simplify this approach by using
early termination instead of (probabilistic) verification, stopping the
loop as soon as several interpolations in a row lead to the same
polynomial. But without any guarantees on the outputs of \Uint{} when
given a too-small bit-length bound $s$, there is no way to analyze the
early termination strategy.

\section{Open question}\label{sec:supersparse}

A natural question, which we leave open, is whether a soft-optimal
algorithm for integer polynomial interpolation with unbalanced
coefficients \emph{and large, unbalanced exponents} is possible.
That is, we have shown an
algorithm that runs in $\bigot{s\log D}$ time; is even better
$\bigot{s}$ possible?
We will outline two possible approaches towards this improvement, and
briefly explain why both do not seem to work with our current
techniques.

At a high level, our algorithm \Uint{} and its subroutine \UintSlice{} is \emph{top-down} in nature: it first retrieves
only the largest terms by making relatively few evaluations at very high precision. Then these are added to the result,
and we proceed to find more terms at the next level, using a greater number of lower-precision evaluations.

The complexity challenge comes at the point of performing evaluations; in particular we need some $p = \bigot{k\log D}$
evaluations at each step in order to retrieve the $k$ largest terms. In prior work such as \cite{GGPR22a}, this extra
$\log D$ factor in evaluations is avoided by using Prony's method as in \citet{BenOrTiwari}, but that doesn't work in
this context because \emph{we do not actually have a $k$-sparse polynomial}! Instead we have a polynomial with many more
than $k$ terms, from which we only want to extract the $k$ largest coefficients. So the Prony method (in exact
arithmetic) cannot be used, and we have to resort to a dense ``over-sampling'' approach and incur the extra $\log D$
factor in cost.
Note that the numerical sparse FFT and related methods \citep{sparse-fft-soda14,sfft-music,Plonka2018} do tackle this
problem of retrieving only the largest coefficients via evaluation/interpolation. But numerical evaluation unfortunately
does not fit with our update step to cancel out the large terms once recovered.
 
A completely opposite approach to ours would be to instead start by recovering all of the \emph{small} terms of the
unknown $f$, then cancel these and iterate, at each step retrieving fewer terms, at higher size and evaluation
precision. This \emph{bottom-up} approach seems at first glance to work very well; in particular it solves the
aforementioned issue of sparsity because, at the point of attempting to recover some $k$ large terms, the difference
polynomial $f-f^*$ really is $k$-sparse, and only $O(k)$ evaluations would be needed at each step.

But unfortunately there is a subtle and seemingly devastating obstacle
to this approach, which is the cost of evaluating the already-recovered
terms $f^*$ at each step. For example, consider an extreme case where
$f$ has roughly $s/\log s$ very small terms with $O(\log s)$ bits each,
and then a constant number of very large terms with $O(s)$ bits each.
After recovering all the small terms, the explicit polynomial $f^*$ has
sparsity $\bigo{s}$ (but small height), and now to find the remaining
few large terms, we need to evaluate this polynomial to very high
precision, that is, with a $s$-bit modulus.

But to our knowledge there
is no known technique to evaluate a large low-height polynomial at just
a few points to very high precision. Instead, the standard approach to
evaluate $f^*(\omega)\bmod m$ with an $s$-bit modulus $m$ would require
first computing $\omega^{e_i}\bmod m$ for each exponent $e_i$ in $f^*$,
which already results in a bit-cost of $\bigot{s^2}$, obviously not
quasi-linear time.

Of course there could be an entirely different approach to achieving
$\bigot{s}$ runtime, but we thought it would be prudent to share these
two (failed) attempts of ours so far, in the hopes of provoking new
ideas and future work.

\section*{Acknowledgement}

  We are grateful to the reviewers for their insightful comments. This work was supported in part by the French National
  Research Agency with the grants: 22-PECY-0010; 22-PECY-003 and ANR-21-CE39-0006.


\newcommand{\Gathen}{\relax}\newcommand{\Hoeven}{\relax}

\end{document}